\documentclass{tlp}

\usepackage{amsmath}
\usepackage{graphicx}
\usepackage{multirow}
\usepackage{xspace}
\usepackage{array}
\usepackage{booktabs}
\usepackage{subcaption}
\usepackage{caption}
\usepackage{amsmath}
\usepackage{amsfonts}
\usepackage{hyperref}
\usepackage{cleveref}
\usepackage{flexisym}
\usepackage{algorithm}
\usepackage[noend]{algpseudocode}
\usepackage{xfrac}

\algnewcommand\algorithmicforeach{\textbf{for each}}
\algdef{S}[FOR]{ForEach}[1]{\algorithmicforeach\ #1\ \algorithmicdo}

\newcommand{\funcname}{\ensuremath{\mathsf{\projenum}}}
\newcommand{\cnt}{\ensuremath{\mathsf{cnt}}}
\newcommand{\Card}[1]{|#1|}

\newcommand{\comp}[1]{\mathsf{Components}(#1)}
\newcommand{\projenu}[1]{\mathsf{ProjMinModels}(#1)}
\newcommand{\var}[1]{\mathsf{Var}(#1)}

\newcommand{\blocked}{\mathcal{B}}
\newcommand{\mmodel}[1]{\mathsf{MinModels}(#1)}
\newcommand{\blockedmmodel}[1]{\mathsf{MinModelswithBlocking}(#1)}
\newcommand{\rules}[1]{\mathsf{Rules}(#1)}
\newcommand{\body}[1]{\mathsf{Body}(#1)}
\newcommand{\at}[1]{\mathsf{atoms}(#1)}
\newcommand{\head}[1]{\mathsf{Head}(#1)}
\newcommand{\answer}[1]{\mathsf{AS}(#1)}
\newcommand{\mingen}[1]{\mathsf{MG}(#1)}

\newcommand{\true}{\ensuremath{\mathsf{true}}\xspace}
\newcommand{\false}{\ensuremath{\mathsf{false}}\xspace}

\newcommand{\dlp}[1]{\mathcal{DLP}(#1)}

\newcommand{\clingo}{clingo\xspace}
\newcommand{\projenum}{\ensuremath{\mathsf{Proj\text{-}Enum}}\xspace}
\newcommand{\approxasp}{ApproxASP\xspace}
\newcommand{\hashcounter}{\ensuremath{\mathsf{HashCount}}\xspace}

\newcommand{\items}{\ensuremath{\mathcal{I}}}
\newcommand{\cover}[1]{\ensuremath{\mathcal{C}(#1)}}
\newcommand{\qes}[1]{\mathsf{TQP}(#1)}
\newcommand{\minc}{\mathsf{C_{min}}}

\newcommand{\approxfunction}{\ensuremath{\mathsf{\hashcounter}}}

\newcommand{\lst}{\mathsf{m^{\star}}}

\newcommand{\pr}[1]{\mathsf{Pr}[#1]}
\newcommand{\bigcell}{\mathsf{hasMinModels}}

\newcommand{\loindex}{\mathsf{loIndex}}
\newcommand{\hiindex}{\mathsf{hiIndex}}
\newcommand{\sstar}{\ensuremath{s^{\star}}}
\newcommand{\lastMMFound}{\mathsf{\hat{m}}}

\newcommand{\Break}{\textbf{break} }
\newcommand{\tint}{\ensuremath{\mathcal{T}}}
\newcommand{\cut}{\ensuremath{\mathcal{C}}}
\newcommand{\support}{\ensuremath{\mathcal{X}}}
\newcommand{\timeout}{\ensuremath{\mathsf{Timeout}} }
\newcommand{\error}{\ensuremath{\mathsf{Error}}}

\newcommand{\p}{\mathsf{P}}
\newcommand{\np}{\mathsf{NP}}
\newcommand{\co}{\mathsf{co}}
\newcommand{\mincount}{\ensuremath{\#\mathsf{MinModels}}\xspace}
\newcommand{\hybrid}{\ensuremath{\mathsf{MinLB}}\xspace}
\newcommand{\computecut}[1]{\ensuremath{\mathsf{Cut}(#1)}}
\newcommand{\is}[1]{\ensuremath{\mathsf{IndependentSupport}(#1)}}

\newtheorem{lemma}{Lemma}
\newtheorem{example}{Example}
\newtheorem{theorem}{Theorem}

\hypersetup{
    colorlinks=true,
    linkcolor=blue,
    filecolor=magenta,
    urlcolor=blue,
    citecolor = blue,
}
\usepackage{hyperref}
\usepackage{xcolor}
\begin{document}

\lefttitle{Kabir and Meel}

\jnlPage{1}{8}
\jnlDoiYr{2021}
\doival{10.1017/xxxxx}

\title[Lower Bounding Minimal Model Count]{On Lower Bounding Minimal Model Count}

\begin{authgrp}
\author{Mohimenul Kabir}
\affiliation{National University of Singapore}
\author{Kuldeep S Meel}
\affiliation{University of Toronto}
\end{authgrp}

\history{\sub{xx xx xxxx;} \rev{xx xx xxxx;} \acc{xx xx xxxx}}

\maketitle

\begin{abstract}
Minimal models of a Boolean formula play a pivotal role in various reasoning tasks. 
While previous research has primarily focused on qualitative analysis over minimal models; our study concentrates on the quantitative aspect, specifically counting of minimal models.
Exact counting of minimal models is strictly harder than $\#\p$, prompting our investigation into establishing a lower bound for their quantity, which is often useful in related applications.
In this paper, we introduce two novel techniques for counting minimal models, leveraging the expressive power of answer set programming: the first technique employs methods from knowledge compilation, while the second one draws on recent advancements in hashing-based approximate model counting. 
Through empirical evaluations, we demonstrate that our methods significantly improve the lower bound estimates of the number of minimal models, surpassing the performance of existing minimal model reasoning systems in terms of runtime.

\end{abstract}

\begin{keywords}
Minimal model, Propositional Circumscription, Model Counting, ASP
\end{keywords}
\section{Introduction}
Given a propositional formula $F$, a model $\sigma \models F$ is {\em minimal} if $\forall \sigma\textprime \subset \sigma$, it holds that $\sigma\textprime \not \models F$~\citep{ABFL2017}. 
Minimal model reasoning is fundamental to several tasks in artificial intelligence, including circumscription~\citep{Mccarthy1980,Lifschitz1985}, 
default logic~\citep{Reiter1980}, diagnosis~\citep{DMR1992}, and deductive databases under the generalized closed-world assumption~\citep{Minker1982}.
Although not new, minimal model reasoning has been the subject of several studies~\citep{EG1993,BD1996,Rachel2005,KK2003}, 
covering tasks such as {\em minimal model finding} (finding a single minimal model), {\em minimal model checking}
(deciding whether a model is minimal), and {\em minimal model entailment and membership} (deciding whether a literal belongs to all minimal models or some minimal models, respectively)~\citep{RP1997}.

Complexity analysis has established that minimal model reasoning is intractable, tractable only for specific subclasses of CNF (Conjunctive Normal Form) formulas. 
Typically, finding one minimal model for positive CNF formulas\footnote{A CNF formula is positive if each clause has at least one positive literal. Every positive CNF formula has at least one minimal model.} is in $\p^{\np[\mathcal{O}(\log{n})]}$-hard~\citep{Cadoli1992a}.
Additionally, checking whether a model is minimal is $\co$-$\np$-complete~\citep{Cadoli1992b}, whereas queries related to entailment and membership are positioned at the second level of the polynomial hierarchy~\citep{EG1993}.

This study delves into a nuanced reasoning task on minimal models, extending beyond the simplistic binary version of decision-based queries. 
Our focus shifts towards quantitative reasoning with respect to minimal models. Specifically, we aim to count the number of minimal models for a given propositional formula. 
While enumerating a single minimal model is insufficient in many applications, counting the number of minimal models provides a useful metric for related measures~\citep{Thimm2016,HS2008}.
Apart from specific structures of Boolean formulas, exact minimal model counting is $\#\co$-$\np$-complete~\citep{KK2003}, established through {\em subtractive reductions}.

Although minimal models can theoretically be counted by iteratively employing minimal model finding oracles, this approach is practical only for a relatively small number of minimal models and becomes impractical as their number increases.
Advanced model counting techniques have scaled to a vast number of models through sophisticated {\em knowledge compilation} methods~\citep{Darwiche2004,Thurley2006}, 
which involve transforming an input formula into a specific representation that enables efficient model counting based on the size of this new representation. 
However, applying knowledge compilation to minimal model counting presents unique challenges, which is elaborated in~\Cref{sec:estimation_methods}. 
Beyond knowledge compilation, approximate model counting has emerged as a successful strategy for estimating the number of models with probabilistic guarantees~\citep{CMV2013}. 
In particular, the {\em hashing-based} technique, which partitions the search space into {\em smaller, roughly equal} partitions using randomly generated XOR constraints~\citep{GSS2021}, \
has attracted significant attention. The model count can be estimated by enumerating the models within one randomly chosen partition~\citep{CMV2013}.

Our empirical study reveals that both approaches to minimal model counting face scalability issues in practical scenarios. Furthermore, knowing a lower bound of the model count is still useful in many applications and is often computed in the model counting literature~\citep{GSS2007}.  
Some applications require the enumeration of all minimal models~\citep{JSS2016,BCGJK2022}, but complete enumeration becomes infeasible for a large number of minimal models. Here, the lower bound of the number of minimal models provides a useful criterion for assessing the feasibility of enumerating all minimal models. 
Knowing this lower bound of the model count is often beneficial to estimate the size of the search space, which enables more specific targeting within the search space~\citep{FGR2022}. 
Consequently, our research shifts focus from counting all minimal models to determining a lower bound for their number.

The primary contribution of this paper is the development of methods to estimate a lower bound for the number of minimal models of a given propositional formula. 
This is achieved by integrating knowledge compilation and hashing-based techniques with minimal model reasoning, thus facilitating the estimation of lower bounds. 
At the core, the proposed methods conceptualize minimal models of a formula as {\em answer sets} of an ASP program; Answer Set Programming (ASP) is a declarative programming paradigm for knowledge representation and reasoning~\citep{MT1999}.
Additionally, our proposed methods depend on the efficiency of well-engineered ASP systems.
Our approach utilizing knowledge compilation effectively counts the number of minimal models or provides a lower bound. 
Besides, our hashing-based method offers a lower bound with a probabilistic guarantee. 
We apply our minimal model counting method to the domain of itemset mining, showcasing its utility. 
The effectiveness of our proposed methods has been empirically validated on datasets from model counting competitions and itemset mining. 
To assess the performance of our proposed methods, we introduce a new metric that considers both the quality of the lower bound and the computational time;
our methods achieve the best score compared to existing minimal model reasoning systems. 

The paper is organized as follows:~\Cref{sec:preliminaries} presents the background knowledge necessary to understand the main contributions of the paper; 
~\Cref{sec:estimation_methods} outlines our proposed techniques for estimating the lower bound on the number of minimal models; 
\Cref{sec:experiment} demonstrates the experimental evaluation of our proposed techniques; and~\Cref{sec:conclusion} concludes our work with some indications of future research directions.  
Due to space constraints, the related work section, implementation details, theoretical analysis of the proposed methods, and part of the experimental analysis are deferred to the appendix.

\section{Preliminaries}
\label{sec:preliminaries}
Before going to the technical description, we present some background about propositional satisfiability, answer set programming, itemset mining from data mining, and a relationship between minimal models and minimal generations in transaction records. 

\paragraph{Propositional Satisfiability.}
In propositional satisfiability, we define the domain $\{0,1\}$, which is equivalently $\{\mathsf{false}, \mathsf{true}\}$ and 
a \emph{propositional variable} or {\em atom} $v$ takes a value from the domain. A \emph{literal} $\ell$ is either a variable $v$ (positive literal) or its
negation $\neg{v}$ (negative literal).
 A \emph{clause} $C$ is a {\em disjunction} of literals, denoted as $C = \bigvee_{i} \ell_i$. 
A Boolean formula  $F$, in \emph{Conjunctive Normal Form (CNF)}, is a {\em conjunction}
of clauses, represented as $F = \bigwedge_{j} C_j$. We use the notation $\var{F}$ to denote the set of variables within  
$F$.

An assignment $\tau$ over $X$ is a function $\tau: X \rightarrow \{0,1\}$, where $X \subseteq \var{F}$.  For an atom $v \in X$, we define
$\tau(\neg{v}) = 1 - \tau(v)$. The assignment $\tau$ over $\var{F}$ is a {\em model} of $F$ if $\tau$ evaluates $F$ to be \true. 
Given $X \subseteq \var{F}$ and an assignment $\tau$, we use the notation $\tau_{\downarrow X}$ to denote the {\em projection} of $\tau$ onto variable set $X \subseteq \var{F}$.
Given a CNF formula $F$ (as a set of clauses) and an assignment
$\tau: X \rightarrow \{0,1\}$, where $X \subseteq \var{F}$, 
the \textit{unit propagation} of $\tau$ on $F$, denoted
$F|_{\tau}$, is recursively defined as follows:

$F|_{\tau} = \begin{cases}
    1 & \text{if $F \equiv 1$}\\
    F'|_{\tau} & \text{if $\exists C \in F$ s.t. $F' = F \setminus \{C\}$,} \text{$\ell \in C$ 
    and $\tau(\ell) = 1$} \\
    F'|_{\tau} \cup \{C'\}  & \text{if $\exists C \in F$ s.t. $F' = F \setminus \{C\}$,} \text{$\ell \in C$,}  \text{$C' = C \setminus \{\ell\}$} 
       \\&\text{ and ($\tau(\ell) = 0$ or $\{\neg{\ell}\} \in F$)}
\end{cases}$

We often consider an assignment $\tau$ as a set of literals it assigns and $\var{\tau}$ denotes the set of variables assigned by $\tau$. For two assignments $\tau_1$ and $\tau_2$, $\tau_1$ satisfies $\tau_2$, denoted as $\tau_1 \models \tau_2$, 
if ${\tau_1}_{\downarrow \var{\tau_2}} = \tau_2$. Otherwise, $\tau_1$ does not satisfy $\tau_2$, denoted as $\tau_1 \not \models \tau_2$. 

An XOR constraint over $\var{F}$ is a Boolean ``XOR'' ($\oplus$) applied to the variables $\var{F}$. 
A random XOR constraint over variables $\{x_1, \ldots, x_k\}$ is expressed as $a_1 \cdot x_1 \oplus \ldots a_k \cdot x_k \oplus b$, where all $a_i$ and $b$ follow the {\em Bernoulli} distribution with a probability of $\sfrac{1}{2}$.
An XOR constraint $x_{i_1} \oplus \ldots x_{i_k} \oplus 1$ (or $x_{i_1} \oplus \ldots x_{i_k} \oplus 0$ resp.) is evaluated as \true if an even (or odd resp.) number of variables from $\{x_{i_1}, \ldots, x_{i_k}\}$ are assigned to \true. 

To define minimal models of a propositional formula $F$, we introduce an {\em ordering operator} over models.
For two given models $\tau_1$ and $\tau_2$, $\tau_1$ is considered {\em smaller} than $\tau_2$, denoted as $\tau_1 \leq \tau_2$, if and only if for each $x \in \var{F}$, $\tau_1(x) \leq \tau_2(x)$. 
We define $\tau_1$ as {\em strictly smaller} than $\tau_2$, denoted as $\tau_1 < \tau_2$, if $\tau_1 \leq \tau_2$ and $\tau_1 \neq \tau_2$.
A model $\tau$ is a {\em minimal model} of $F$ if and only if $\tau$ is a model of $F$ and no model of $F$ is strictly smaller than $\tau$. We use the notation $\mmodel{F}$ to denote minimal models of $F$ and 
for a set $X \subseteq \var{F}$, $\mmodel{F}_{\downarrow X}$ denotes the minimal models of $F$ projected onto the variable set $X$.
The minimal model counting problem seeks to determine the cardinality of $\mmodel{F}$, denoted $\Card{\mmodel{F}}$.

In this paper, we sometimes represent minimal models by listing the variables assigned as \true. For 
example, suppose $\var{F} = \{a,b,c\}$ and under minimal model $\tau = \{a, b\}$,  
$\tau(a) = \tau(b) = \true$ and $\tau(c) = \false$. The notation $\neg{\tau}$ denotes the negation of assignment $\tau$; in fact, $\neg{\tau}$ is a clause or disjunction of literals (e.g., when $\tau = \{a, b\}$, $\neg{\tau} = \neg{a} \vee \neg{b}$). 
Throughout the paper, we use the notations $\tau$ and $\sigma$ to denote an arbitrary assignment and a minimal model of $F$, respectively.
For each model $\sigma \in \mmodel{F}$, each of the variables assigned to \true is {\em justified}; more specifically, 
for every literal $\ell \in \sigma$, there exists a clause $c \in F$ such that $\sigma \setminus \{\ell\} \not \models c$. Otherwise, 
$\sigma \setminus \{\ell\} < \sigma$, is a model of $F$.

\paragraph{Answer Set Programming.}
An \textit{answer set program} $P$ consists of a set of rules, each rule is structured as follows:
\begin{align}
\label{eq:general_rule}
\text{Rule $r$:~~}a_1 \vee \ldots a_k \leftarrow b_1, \ldots, b_m, \textsf{not } c_1, \ldots, \textsf{not } c_n
\end{align}
where, $a_1, \ldots, a_k, b_1, \ldots, b_m, c_1, \ldots, c_n$ are propositional variables or atoms, and $k,m,n$ are non-negative integers. 
The notations $\rules{P}$ and $\at{P}$ denote the rules and atoms within the program $P$. 
In rule $r$, the operator ``\textsf{not}'' denotes \textit{default negation}~\citep{clark1978}. For each 
rule $r$ (\cref{eq:general_rule}), we adopt the following notations: the atom set $\{a_1, \ldots, a_k\}$ constitutes the {\em head} of $r$, denoted by $\head{r}$, the set $\{b_1, \ldots, b_m\}$ is referred to as the {\em positive body atoms} of $r$, denoted by $\body{r}^+$, and the set $\{c_1, \ldots, c_n\}$ is referred to as the \textit{negative body atoms} of $r$, denoted by $\body{r}^-$.
A rule $r$ called a {\em constraint} when $\head{r}$ contains no atom.
A program $P$ is called a {\em disjunctive logic program} if $\exists r \in \rules{P}$ such that $\Card{\head{r}} \geq 2$~\citep{BD1994}.

In ASP, an interpretation $M$ over $\at{P}$ specifies which atoms are assigned \true; that is, an atom $a$ is \true under $M$ if and only if $a \in M$ (or \false when $a \not\in M$ resp.). 
An interpretation $M$ satisfies a rule $r$, denoted by $M \models r$, if and only if $(\head{r} \cup \body{r}^{-}) \cap M \neq \emptyset$ or $\body{r}^{+} \setminus M \neq \emptyset$. An interpretation $M$ is a {\em model} of $P$, denoted by $M \models P$, when $\forall_{r \in \rules{P}} M \models r$. 
The \textit{Gelfond-Lifschitz (GL) reduct} of a program $P$, with respect to an interpretation $M$, is defined as $P^M = \{\head{r} \leftarrow \body{r}^+| r \in \rules{P}, \body{r}^- \cap M = \emptyset\}$~\citep{GL1991}.
An interpretation $M$ is an {\em answer set} of $P$ if $M \models P$ and no $M\textprime \subset M$ exists such that $M\textprime \models P^M$.
We denote the answer sets of program $P$ using the notation $\answer{P}$.

\paragraph{From Minimal Models to Answer Sets.}
Consider a Boolean formula, $F = \bigwedge_{i} C_i$, where each clause is of the form: $C_i = \ell_0 \vee \ldots \ell_{k} \vee \neg{\ell_{k+1}} \vee \ldots \neg \ell_{m}$.
We can transform each clause $C_i$ into a rule $r$ of the form: $\ell_{0} \vee \ldots \vee \ell_{k} \leftarrow \ell_{k+1}, \ldots, \ell_{m}$. Given a formula $F$, 
let us denote this transformation by the notation $\dlp{F}$. Each minimal model of $F$ corresponds uniquely to an answer set of $\dlp{F}$ (the proof is deferred to the appendix).

\paragraph{Approximate Lower Bound.}
We denote the probability of an event $e$ using the notation $\pr{e}$.
For a Boolean formula $F$, let $c$ represents a lower bound estimate for the number of minimal models of $F$. We assert that 
$c$ is a lower bound for the number of minimal models with a {\em confidence} $\delta$, when 
$\pr{c \leq \Card{\mmodel{F}}} \geq 1 - \delta$.

\paragraph{Minimal Generator in Itemset Mining.}
We define transactions over a finite set of items, denoted by $\items$. A \textit{transaction} $t_i$ is an ordered pair of $(i, I_i)$, where $i$ is the unique identifier of the transaction and $I_i \subseteq \items$ represents the set of items involved in the transaction. 
A transaction database is a collection of transactions, where each uniquely identified by the identifier $i$, corresponding to the transaction $t_i$. A transaction $(i, I_i)$ supports an itemset $J \subseteq \mathcal{I}$ if $J \subseteq I_i$.
The \textit{cover} of an itemset $J$ within a database $D$, denoted as $C(J,D)$, is defined as: $\cover{J,D} = \{i | (i,I_i) \in D \text{ and } J \subseteq I_i\}$.
Given an itemset $I$ and transaction database $D$, the itemset $I$ is a \textit{minimal generator} of $D$ if, for every itemset $J$ where $J \subset I$, it holds that $\cover{I,D} \subset \cover{J,D}$.

\paragraph{Encoding Minimal Generators as Minimal Models.}
Given a transaction database $D$, we encode a Boolean formula $\mingen{D}$ such that minimal models of $\mingen{D}$ correspond one-to-one with the minimal generators of $D$.
This encoding introduces two types of variables: (i) for each item $a \in \items$, we introduce a variable $p_a$ to denote that $a$ is present in a minimal generator  
(ii) for each transaction $t_i$, we introduce a variable $q_i$ to denote the presence of the itemset in the transaction $t_i$. Given a transaction database $D = \{t_i| i = 1,\ldots n\}$, consisting of the union of transactions $t_i$,
consider the following Boolean formula:
\begin{align}
    \mingen{D} = \bigwedge_{i=1}^{n} \bigg( \neg{q_i} \rightarrow \bigvee_{a \in \items \setminus I_i} p_a \bigg)
\end{align}
\begin{lemma}
Given a transaction database $D$, $\sigma$ is a minimal model of $\mingen{D}$ if and only if the corresponding itemset $I_{\sigma} = \{a | p_a \in \sigma\}$ is a minimal generator of $D$.
\label{lemma:minimal_model_to_minimal_generator}
\end{lemma}
The encoding of $\mingen{D}$ bears similarities to the encoding detailed in~\citep{JSS2017,Salhi2019}. However, our encoding achieves compactness by incorporating a one-sided implication, which enhances the efficiency of the representation.
The proof of~\Cref{lemma:minimal_model_to_minimal_generator} is deferred to the appendix.

\section{Related Works}
\label{section:review}
Given its significance in numerous reasoning tasks, minimal model reasoning has garnered considerable attention from the scientific community.

Minimal models of a Boolean formula $F$ can be computed using an iterative approach with a SAT solver~\citep{LYZR2021,MPM2015,MHJPB2013}. The fundamental principle is as follows: for any model 
$\alpha \in \mmodel{F}$, no model of $F$ can exist that is strictly smaller than $\alpha$; thus, $F \wedge \neg{\alpha}$ yields no model. 
Conversely, if $F \wedge \neg{\alpha}$ returns a model, it must be strictly smaller than $\alpha$.

Minimal models can be efficiently determined using {\em unsatisfiable core-based} MaxSAT algorithms~\citep{alviano2017}. 
This technique leverages the unsatisfiable core analysis commonly used in MaxSAT solvers and operates within an incremental solver to enumerate minimal models sorted by their size.
In parallel, another line of research focuses on the enumeration of minimal models by applying cardinality constraints to calculate models of bounded size~\citep{LS2008,FVCG2016}. 
Notably, Faber et al.~\citeyear{FVCG2016} employed an algorithm that utilized an external solver for the enumeration of cardinality-minimal models of a given formula.
Upon finding a minimal model, a blocking clause is integrated into the input formula, ensuring that these models are not revisited by the external solver.

There exists a close relationship between minimal models of propositional formula and answer sets of ASP program~\citep{BD1994,LL2006}. 
Beyond solving disjunctive logic programs (ref.~\Cref{sec:preliminaries}), minimal models can also be effectively computed using specialized techniques within the context of ASP, such as domain heuristics~\citep{GKRR2013} and preference relations~\citep{BDRS2015}.

Due to the intractability of minimal model finding, 
research has branched into exploring specific subclasses of positive CNF formulas where minimal models can be efficiently identified within polynomial time~\citep{BD1996, ABFL2017}. 
Notably, a Horn formula possesses a singular minimal model, which can be derived in linear time using unit propagation~\citep{BD1996}.
Rachel and Luigi~\citeyear{RP1997} developed an {\em elimination algorithm} designed to find and verify minimal models for {\em head-cycle-free} formulas. 
Fabrizio et al.~\citeyear{ABFL2022} introduced the {\em Generalized Elimination Algorithm} (GEA), capable of identifying minimal models across any positive formula when paired with a suitably chosen eliminating operator.
The efficiency of the GEA hinges on the complexity of the specific eliminating operator used. With an appropriate eliminating operator, the GEA can determine minimal models of head-elementary-set-free CNF formulas in polynomial time. Notably, this category is a broader superclass of the head-cycle-free subclass.

Graph-theoretic properties have been effectively utilized in the reasoning about minimal models. 
Specifically, Fabrizio et al.~\citeyear{ABFL2022} demonstrated that minimal models of positive CNF formulas can be decomposed based on the structure of their dependency graph. 
Furthermore, they introduced an algorithm that leverages model decomposition, utilizing the underlying dependency graph to facilitate the discovery of minimal models. 
This approach underscores the utility of graph-theoretic concepts in enhancing the efficiency and understanding of minimal model reasoning.

To the best of our knowledge, the literature on minimal model counting is relatively sparse. 
The complexity of counting minimal models for specific structures of Boolean formulas, such as Horn, dual Horn, bijunctive, and affine, has been established as $\#\p$~\citep{DH2008}. 
This complexity is notably lower than the general case complexity, which is $\#\co$-$\np$-complete~\citep{KK2003}.

\section{Estimating the Number of Minimal Models}
\label{sec:estimation_methods}
In this section, we introduce methods for determining a lower bound for the number of minimal models of a Boolean formula. 
We detail two specific approaches aimed at estimating this number. 
The first method is based on the {\em decomposition} of the input formula, whereas the second method utilizes a hashing-based approach of approximate model counting. 

\subsection{Formula Decomposition and Minimal Model Counting}
\label{subsec:projenum}
Considering a Boolean formula $F = F_1 \wedge F_2$, we define the components $F_1$ and $F_2$ as {\em disjoint} 
if no variable of $F$ is mentioned by both components $F_1$ and $F_2$ (i.e., $\var{F_1} \cap \var{F_2} = \emptyset$). 
Under this condition, the models of $F$ can be independently derived from the models of $F_1$ and $F_2$
 and the vice versa. 
Thus, if $F_1$ and $F_2$ are disjoint in the formula $F = F_1 \wedge F_2$, 
the total number of models of $F$ is the product of the number of models of $F_1$ and $F_2$. This principle underpins the decomposition technique frequently applied in knowledge compilation~\citep{LM2017}.

Building on the concept of the knowledge compilation techniques, we introduce a strategy centered on formula decomposition to count minimal models. 
Unlike methods that count models for each disjoint component, we enumerate minimal models of $F$ projected onto the variables of disjoint components. Our approach incorporates a level of enumeration that stops upon enumerating a specific count of minimal models, thereby providing a lower bound estimate of the total number of minimal models.
Our method utilizes a straightforward ``Cut'' mechanism to facilitate formula decomposition.

\paragraph{Formula Decomposition by ``Cut''}
A ``cut'' $\cut$ within a formula $F$ is identified as a subset of $\var{F}$ such that for every assignment $\tau \in 2^{\cut}$, 
$F|_{\tau}$ {\em effectively decomposes} into disjoint components~\citep{LM2017}. This concept is often used in context of model counting~\citep{KJ2021}. 
It is important to note that models of $F|_{\tau}$ can be directly expanded into models of $F$. 

\paragraph{Challenges in Knowledge Compilation for Counting Minimal Models}
When it comes to counting minimal models, the straightforward application of unit propagation and the conventional decomposition approach are not viable.
More specifically, simple unit propagation does not preserve minimal models. Additionally, the count of minimal models cannot be simply calculated by multiplying the counts of minimal models of its disjoint components. 
An example provided subsequently demonstrates these inconsistencies.

\begin{example}
\label{example:example1}
Consider a formula $F = \{a \vee b \vee c, \neg{a} \vee \neg{b} \vee d, \neg{a} \vee \neg{b} \vee e\}$.\\
(i)~With the assignment $\tau_1 = \{e\}$. Then $\{a\}$ becomes a minimal model of $F|_{\tau_1}$. However, the extended assignment $\tau_1 \cup \{a\}$ is not a minimal 
model of $F$.\\
(ii)~Considering a cut $\cut = \{a, b\}$ and the partial assignment $\tau_2 = \{a, b\}$, then 
$F|_{\tau_2}$ is decomposed into two components, each containing the unit clauses $\{d\}$ and $\{e\}$, respectively.
Despite this, the combined assignment $\tau_2 \cup \{d, e\} = \{a, b, d, e\}$ is not a minimal model of $F$, as  
a strictly smaller assignment $\{a, d, e\}$ also satisfies $F$.
\end{example}

Traditional methods such as unit propagation and formula decomposition cannot be straightforwardly applied to minimal model counting. Importantly, every atom in a minimal model must be justified. 
In~\Cref{example:example1}, (i)~the variable $e$ is assigned truth values without justification, leading to an incorrect minimal model when the assignment is extended with the minimal model of $F|_{\{e\}}$. 
(ii)~The formula is decomposed without justifying the variables $a$ and $b$, resulting in incorrect minimal model when combining assignments from the other two components of $F|_{\{a,b\}}$.
Therefore, for accurate minimal model counting, operations such as unit propagation and formula decomposition must be applied only to assignments that are justified.
Consequently, a knowledge compiler for minimal model counting must frequently verify the justification of assignments. It is worth noting that verifying the justification of an assignment is computationally intractable.

\paragraph{Minimal Model Counting using Justified Assignment}
We introduce the concept of a {\em justified assignment} $\tau^{\star}$, based on a given assignment $\tau$. 
Within the minimal model semantics, any assignment of \false is inherently justified.
Therefore, we define justified assignment $\tau^{\star}$ as follows:
    $\tau^{\star} = \tau_{\downarrow \{v \in \var{F} | \tau(v) = 0\}}$.

By applying unit propagation of $\tau^{\star}$, instead of $\tau$, every minimal model derived from $F|_{\tau^{\star}}$ can be seamlessly extended into a 
minimal model of $F$. 
While $F|_{\tau}$ effectively decompose into multiple disjoint components, the use of a justified assignment $\tau^{\star}$ does not necessarily lead to effectively $F|_{\tau^{\star}}$ decomposing into disjoint components. 

A basic approach to counting minimal models involves enumerating all minimal models.
When a formula is decomposed into multiple disjoint components, the number of minimal models can be determined by conducting a projected enumeration over these disjoint variable sets and subsequently multiplying the counts of projected minimal models. 
The following corollary outlines how projected enumeration can be employed across disjoint variable sets to accurately count the number of minimal models.

\begin{lemma}
\label{theorem:decomposition}
Let $F$ be decomposed into disjoint components $F_1, \ldots, F_k$, with each component $F_i$ having a variable set $V_i = \var{F_i}$, for $i \in [1,k]$. Suppose $V = \var{F} = \bigcup_{i} V_i$. 
Then, 
$\Card{\mmodel{F}} = \prod_{i=1}^k \Card{\mmodel{F}_{\downarrow V_i}}$.
\end{lemma}

\paragraph{Algorithm: Counting Minimal Models by Projected Enumeration}
We introduce an enumeration-based algorithm, called \projenum, that leverages justified assignments and projected enumeration to accurately count the number of minimal models of a Boolean formula $F$.
The algorithm takes in as input a Boolean formula $F$ and a set of variables $\cut$, referred to as a ``cut'' in our context. 
To understand how the algorithm works, we introduce two new concepts: $\blockedmmodel{F, \blocked}$ and $\projenu{F, \tau, X}$.
$\blockedmmodel{F, \blocked}$ finds $\sigma \in \mmodel{F}$ such that $\forall \tau \in \blocked, \sigma \not \models \tau$; here, 
$\blocked$ is a set of {\em blocking} clauses and each blocking clause is an assignment $\tau$.
$\projenu{F, \tau, X}$ enumerates the set $\{\sigma_{\downarrow X} | \sigma \in \mmodel{F}, \sigma \models \tau\}$, where $\tau$ serves as the {\em conditioning} factor and $X$ serves as the projection set. 
The algorithm iteratively processes minimal models of $F$ (Line~\ref{line:whileloop}), starting with an initially empty set of blocking clauses ($\blocked = \emptyset$).
Upon identifying a minimal model $\sigma$, the algorithm projects $\sigma$ onto $\cut$, denoting the projected set as $\tau$. 
Subsequently,~\Cref{alg:projected_enumeration} enumerates all minimal models $\sigma \in \mmodel{F}$ that satisfy $\sigma \models \tau$. 

To address the inefficiency associated with brute-force enumeration, the algorithm utilizes the concept of justified assignment and projected enumeration (Line~\ref{line:projection}).
~\Cref{theorem:decomposition} establishes that the number of minimal models can be counted through multiplication.
The notation $\comp{F}$ (Line \ref{line:forloop}) denotes all disjoint components of the formula $F$. 
It is important to note that if $F|_{\tau^{\star}}$ does not decompose into more than one component, then the projection variable set $X$ defaults to $\var{F}$, 
which leads to brute-force enumeration of non-projected minimal models. 
Finally, the algorithm adds $\tau$ to $\blocked$ (Line~\ref{line:blocking}) to prevent the re-enumeration of the same minimal models. 

\begin{algorithm} \caption{\funcname(F, \cut)}
    \label{alg:projected_enumeration}
    \begin{algorithmic}[1]
    \State $\cnt \gets 0$, $\blocked \gets \emptyset$
    \While {$\exists \sigma \in \blockedmmodel{F, \blocked}$}\label{line:whileloop}
    \State $\tau \gets \sigma_{\downarrow \cut}$, $d \gets 1$
    \ForEach {$\mathsf{comp} \in \comp{F|_{\tau^{\star}}}$}\label{line:forloop}
    \State $d = d \times \Card{\projenu{F, \tau, \var{\mathsf{comp}}}}$\label{line:projection}
    \EndFor
    \State $\cnt \gets \cnt + d$
    \State $\blocked.\mathsf{add}(\tau)$\label{line:blocking}
    \EndWhile
    \State \Return{$\cnt$}
    \end{algorithmic}
\end{algorithm}
\paragraph{Implementation Details of \projenum.}
We implemented \projenum using Python. 
To find minimal models using $\blockedmmodel{F, \blocked}$, the algorithm invokes an ASP solver on $\dlp{F}$. Each assignment $\tau \in \blocked$ is incorporated as a constraint~\citep{ADFPR2022}, which ensures that minimal models of $F$ are preserved~\cite[see Corollary.~2]{KESHFM2022}.
To compute $\projenu{F, \tau, X}$,~\Cref{alg:projected_enumeration} employs an ASP solver on $\dlp{F}$, using $X$ as the projection set. Additionally, it incorporates each literal from $\tau$ as a {\em facet} into $\dlp{F}$~\citep{ARS2018} to ensure that the condition $\tau$ is satisfied.
We noted that the function $\projenu{F, \tau, X}$ requires more time to enumerate all minimal models. To leverage the benefits of decomposition, we enumerate upto a specific threshold number of minimal models (set the threshold to $10^6$ in our experiment) invoking $\projenu{F, \tau, X}$.  
Consequently, our prototype either accurately counts the total number of minimal models or provides a lower bound. 
Employing a tree decomposition technique~\citep{HS2018}, we calculated a cut of the formula that effectively decompose the input formula into several components.

\subsection{Hashing-based Minimal Model Counting}
\label{subsec:hashcount}

The number of minimal models can be approximated using a hashing-based model counting technique, which adds constraints that restrict the search space. Specifically, this method applies {\em uniform and random} XOR constraints to a formula 
$F$, focusing the search on a smaller subspace~\citep{GSS2021}. 
A particular XOR-based model counter demonstrates that 
if $t$ trials are conducted where $s$ random and uniform XOR constraints are added each time, and the constrained formula of $F$ 
is satisfiable in all $t$ cases, then $F$ has at least $2^{s - \alpha}$ models with high confidence, where $\alpha$ is the {\em precision slack}~\citep{GSS2006}. 
Each XOR constraint incorporates variables from $\var{F}$. Our approach to minimal model counting fundamentally derives from the strategy of introducing random and uniform XOR constraints to the formula.
In the domain of approximate model counting and sampling, the XOR constraints consist of variables from a subset of $\var{F}$, denoted as $\support$ within our algorithm, which is widely known as \textit{independent support}~\citep{CMV2016, SM2022}.

\Cref{alg:lower_bound_on_minimal_models_counting} outlines a hashing-based algorithm, named $\hashcounter$, for determining the lower bound of minimal models of a Boolean formula $F$. 
This algorithm takes in a Boolean formula $F$, an independent support $\support$, and a confidence parameter $\delta$.
During its execution, the algorithm generates total $\Card{\support} - 1$ random and uniform XOR constraints, denoted as $Q^i$, where $i$ ranges from $1$ to $\Card{\support} - 1$.
To better explain the operation of the algorithm, we introduce a notation: $\mmodel{F^m}$ represents the minimal models of $F$ satisfying first $m$ XOR constraints, $Q^1, \ldots, Q^m$. 
Upon generating random and uniform XOR constraints, the algorithm finds the value of $m$ such that $\Card{\mmodel{F^m}} > 0$ (meaning that $\exists \sigma \in \mmodel{F}, \sigma \models Q^1 \wedge \ldots \wedge Q^m$), 
while $\Card{\mmodel{F^{m+1}}} = 0$ (meaning that $\not \exists \sigma \in \mmodel{F}, \sigma \models Q^1 \wedge \ldots \wedge Q^{m+1}$) by iterating a loop (Line~\ref{line:while_loop}). 
The loop terminates either when a $\timeout$ occurs or when it successfully identifies the value of $m$. If the $\timeout$ happens, 
the algorithm assigns the maximum observed value of $m$ (denoted as $\lastMMFound$) to $\lst$, ensuring that $\Card{\mmodel{F^{\lastMMFound}}} \geq 1$ (Line~\ref{line:max_value_of_m}), which is sufficient to offer lower bounds. Finally,~\Cref{alg:lower_bound_on_minimal_models_counting} returns $2^{\lst - \alpha}$ as the probabilistic lower bound of $\Card{\mmodel{F}}$.

\begin{algorithm}[h!] \caption{\approxfunction($F, \support, \delta$)}
    \label{alg:lower_bound_on_minimal_models_counting}
    \begin{algorithmic}[1]
        \State $\alpha \gets -\log_2{(\delta)}+1$
        \State generate $\Card{\support} - 1$ random constraints, namely $Q^1, \ldots, Q^{\Card{\support} - 1}$
        \State $\bigcell[0] \gets 1, \bigcell[\Card{\support}] \gets 0, \loindex \gets 0, \hiindex \gets \Card{\support}, m \gets 1$
        \State $\lastMMFound \gets \bot, \lst \gets \bot$
        \For {$i \gets 1$ to $\Card{\support} - 1$} $\bigcell[i] \gets \bot$
        \EndFor
        \While {$\true$} \label{line:while_loop}
        \If {$\timeout$} 
        $\lst \gets \lastMMFound$ \label{line:max_value_of_m} \Break
        \EndIf
        \If {$\exists \sigma \in \mmodel{F^m}$} 
        \State $\lastMMFound \gets \mathsf{Max}(\lastMMFound,m)$
        \If {$\bigcell[m+1] = 0$} 
        $\lst \gets m$
        \Break
        \EndIf
        \For {$i \gets 1$ to $m$} $\bigcell[i] \gets 1$
        \EndFor
        \State $\loindex \gets m$
        \If {$2 \times m < \Card{\support}$}  
        $m \gets 2 \times m$
        \Else 
        \text{ }$m \gets \frac{(\hiindex + m)}{2}$
        \EndIf
        \Else 
        \If {$\bigcell[m-1] = 1$} 
        $\lst \gets m-1$
        \Break
        \EndIf
        \For {$i \gets m$ to $\Card{\support} - 1$} $\bigcell[i] \gets 0$
        \EndFor
        \State $\hiindex \gets m$
        \State $m \gets \frac{(\loindex + m)}{2}$
        \EndIf
        \EndWhile
    \State \Return{$2^{\lst - \alpha}$} \label{line:return_case_2}
    \end{algorithmic}
\end{algorithm}

\paragraph{Implementation Details of \hashcounter}
The effectiveness of an XOR-based model counter is dependent on the performance of a theory+XOR solver~\citep{SGM2020}. 
In our approach, we begin by transforming a given formula $F$ into a disjunctive logic program $\dlp{F}$ and introduce random and uniform XOR constraints into $\dlp{F}$ to effectively partition the minimal models of $F$.
To verify the presence of any models in the XOR-constrained ASP program, we leverage the ASP+XOR solver capabilities provided by ApproxASP~\citep{KESHFM2022}. 
To compute an independent support for \hashcounter, we implemented a prototype inspired by Arjun~\citep{SM2022}, which checks answer sets of a disjunctive logic program in accordance with Padoa’s theorem~\citep{Padoa1901}.

\subsection{Putting It All Together}
\begin{algorithm}[h!] \caption{$\mathsf{MinLB}(F, \delta$)}
    \label{alg:hybrid_algorithm}
    \begin{algorithmic}[1]
    \If {$\not \exists \sigma \in \mmodel{F}$}
    \Return $0$
    \EndIf
    \If {$\Card{\computecut{F}} \leq 50$}
        \Return{\funcname($F, \computecut{F}$)} \label{line:run_projenum}
    \Else
    \text{ } \Return {\approxfunction($F, \is{F}, \delta$)} \label{line:run_hashcount}
    \EndIf
    \end{algorithmic}
\end{algorithm}
We designed a hybrid solver \hybrid, presented in \Cref{alg:hybrid_algorithm}, that selects either \projenum or \hashcounter depending on the decomposability of the input formula. 
The core principle of \hybrid is that \projenum effectively leverages decomposition and projected enumeration on {\em easily decomposable} formulas.
We use the size of the cut ($\Card{\computecut{F}}$) as a proxy to measure the decomposability.
Thus, if $\Card{\computecut{F}}$ is small, then \hybrid employs \projenum on $F$ (Line~\ref{line:run_projenum}); otherwise, it employs \hashcounter (Line~\ref{line:run_hashcount}). 
\begin{theorem}
\label{thm:main_theorem}
    $\pr{\hybrid(F,\delta) \leq \Card{\mmodel{F}}} \geq 1 - \delta$
\end{theorem}

\section{Experimental Results}
\label{sec:experiment}

\paragraph{Benchmarks and Baselines}
Our benchmark set is collected from two different domains: (i) model counting benchmarks from recent competitions~\citep{FHH2020} and (ii) minimal generators benchmark from itemset mining dataset CP4IM\footnote{\url{https://dtai.cs.kuleuven.be/CP4IM/datasets/}}.
We used existing systems for minimal model reasoning as baselines. 
These included various approaches such as (i) repeated invocations of the SAT solver MiniSAT~\citep{LYZR2021}, (ii) the application of MaxSAT techniques~\citep{alviano2017}, (iii) domain-specific heuristics~\citep{GKRR2013}, and (iv) solving $\dlp{F}$ with ASP solvers, all systems primarily count via enumeration.
These systems either return the number of minimal models by enumerating all of them or a lower bound of the number of minimal models in cases where they run out of time or memory.
Experimentally, we observed that, for enumerating all minimal models, the technique solving disjunctive ASP programs using clingo~\citep{GKS2012} surpassed the other techniques. Therefore, we have exclusively reported the performance of clingo in our experimental analysis.
Additionally, we evaluated ApproxASP~\citep{KESHFM2022}, which offers $(\epsilon, \delta)$-guarantees in counting minimal models. We attempted an exact minimal model counting tool using a subtractive approach --- subtracting the {\em non-minimal model} count from the total model count—we denote the implementation using the notation \mincount in further analysis.
In our experiment, we ran \hashcounter with a confidence $\delta$ value of $0.2$ and ApproxASP with a confidence $\delta$ of $0.2$ and tolerance $\epsilon$ of $0.8$.
Note that we cannot compare with $\#$SAT-based exact answer set counter as $\dlp{F}$ is not necessarily a normal logic program~\cite{KCM2024}.

\paragraph{Environmental Settings}
All experiments were conducted on a high-performance computing cluster equipped with nodes featuring AMD EPYC $7713$ CPUs, each with $128$ real cores. 
Throughout the experiment, the runtime and memory limits were set to $5000$ seconds and $16$GB, respectively, for all considered tools.

\paragraph{Evaluation Metric}
The goal of our experimental analysis is to evaluate various minimal model counting tools based on their runtime and the quality of their lower bounds. 
To effectively assess these systems, it is essential to employ a metric that encompasses both runtime performance and the quality of the lower bound. 
Consequently, following the TAP score~\citep{APM2021,KM2023}, we have introduced a metric, called the \textit{Time Quality Penalty} (TQP) score, which is defined as follows:
\[
  \qes{t, C} = 
  \begin{cases}
    2 \times \tint, & \text{if no lower bound is returned}   \\
    t + \tint \times \frac{1 + \log{(\minc + 1)}}{1 + \log{(C+1)}}, & \text{otherwise } \\
  \end{cases}
\]
In the metric, $\tint$ represents the timeout for the experiment, $t$ denotes the runtime of a tool, $C$ is the lower bound returned by that tool,  
and $\minc$ is the minimum lower bound returned by any of the tools under consideration for the instance.
The TQP score is based on the following principle: lower runtime and higher lower bound yield a better score.

\subsection{Performance on Model Counting Competition Benchmark}
We present the TQP scores of \hybrid alongside other tools in~\Cref{table:qes_score_min_counting}. 
This table indicates that \hybrid achieves the lowest TQP scores. 
Among existing minimal model enumerators, clingo demonstrates the best performance in terms of TQP score. 
Additionally, in~\Cref{fig:minimal_model_counting_across_clingo}, we graphically compare the lower bounds returned by \projenum and \hashcounter against those returned by \clingo. 
Here, a point $(x, y)$ indicates that, for an instance, the lower bounds returned by our prototypes and clingo are $2^x$ and $2^y$, respectively. 
For an instance, if the corresponding point resides below the diagonal line, it indicates that \projenum (\hashcounter, resp.) returns a better lower bound than \clingo.
These plots clearly illustrate that \projenum and \hashcounter return better lower bounds compared to other existing minimal model enumerators.

\begin{table}[h]
    \centering
    \begin{tabular}{m{5em} m{5em} m{6em} m{10em}} 
    \toprule
    \clingo & \approxasp & \mincount & \hybrid \footnotesize{(our prototype)}\\
    \midrule
    6491 & 6379 & 7743 & 5599\\
    \bottomrule
    \end{tabular}
    \caption{The TQP scores of \hybrid and other tools on model counting benchmark.}
    \label{table:qes_score_min_counting}
\end{table}

\begin{figure*}[h!]
    \centering
        \begin{subfigure}[t]{0.48\textwidth}
            \centering
            \includegraphics[width=0.6\linewidth]{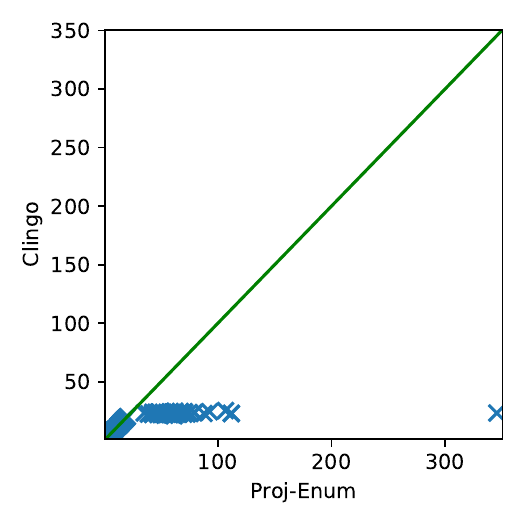}
            \caption{The lower bound returned by \projenum
            }
        \label{fig:power_d1_min_gen}
        \end{subfigure}
        \begin{subfigure}[t]{0.48\textwidth}
            \centering
            \includegraphics[width=0.6\linewidth]{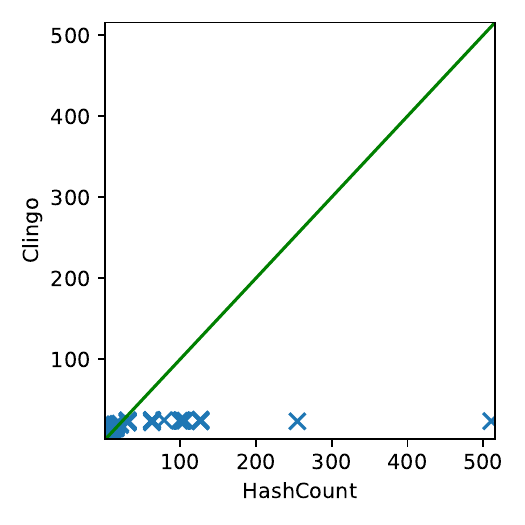}
            \caption{The lower bound returned by \hashcounter
            }
        \label{fig:power_mbound_min_gen}
        \end{subfigure}
    \caption{The lower bound of \projenum and \hashcounter vis-a-vis the lower bound returned by 
    \clingo on minimal model counting benchmark. The axes are in log scale.}
    \label{fig:minimal_model_counting_across_clingo}
\end{figure*}

\subsection{Performance on Minimal Generator Benchmark}
\Cref{table:qes_score_min_generators} showcases the TQP scores of \hybrid alongside other tools on the minimal generator benchmark. 
Notably, \hashcounter achieves the most favorable TQP scores on the benchmark. 
Additionally,~\Cref{fig:minimal_generator_counting_across_clingo} graphically compares the lower bounds returned by \projenum and \hashcounter against those computed by \clingo.  
\begin{table}[h]
    \centering
    \begin{tabular}{m{5em} m{5em} m{6em} m{10em}} 
    \toprule
    \clingo & \approxasp & \mincount & \hybrid \footnotesize{(our prototype)}\\
    \midrule
    6944 & 5713 & 9705 & 5043\\
    \bottomrule
    \end{tabular}
    \caption{The TQP scores of \hybrid and other tools on minimal generator benchmark.}
    \label{table:qes_score_min_generators}
\end{table}

\begin{figure*}[h!]
    \centering
    \begin{subfigure}[t]{0.48\textwidth}
        \centering
        \includegraphics[width=0.6\linewidth]{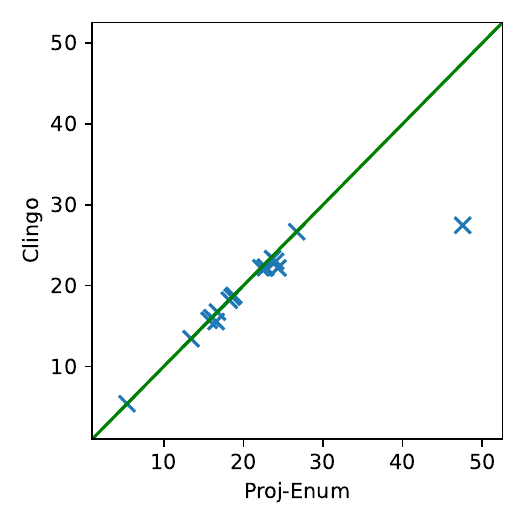}
        \caption{The lower bound returned by \projenum
        }
    \label{fig:power_d1_min_gen}
    \end{subfigure}
    \begin{subfigure}[t]{0.48\textwidth}
        \centering
        \includegraphics[width=0.6\linewidth]{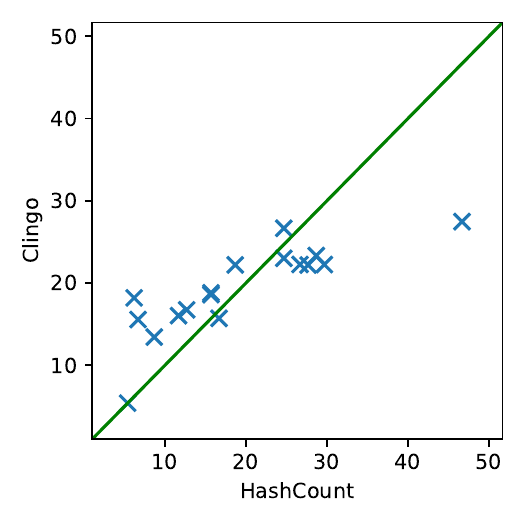}
        \caption{The lower bound returned by \hashcounter
        }
    \label{fig:power_mbound_min_gen}
    \end{subfigure}
    \caption{The lower bound returned by \projenum and \hashcounter vis-a-vis the lower bound given by 
    \clingo on minimal generators benchmark. The axes are in log scale.}
    \label{fig:minimal_generator_counting_across_clingo}
\end{figure*}

For a visual representation of the lower bounds returned by our \hashcounter and \projenum, we illustrate them graphically in~\Cref{fig:compare_count_of_minimal_models}.
In the plot, a point $(x, y)$ signifies that a tool returns a lower bound of at most $2^y$ for $x$ instances.
The plot demonstrates that the lower bounds returned by \projenum and \hashcounter surpass those of existing systems. The more experimental analysis is deferred to the appendix.
\begin{figure*}[h!]
    \centering
        \begin{subfigure}[t]{0.48\textwidth}
            \centering
            \includegraphics[width=0.8\linewidth]{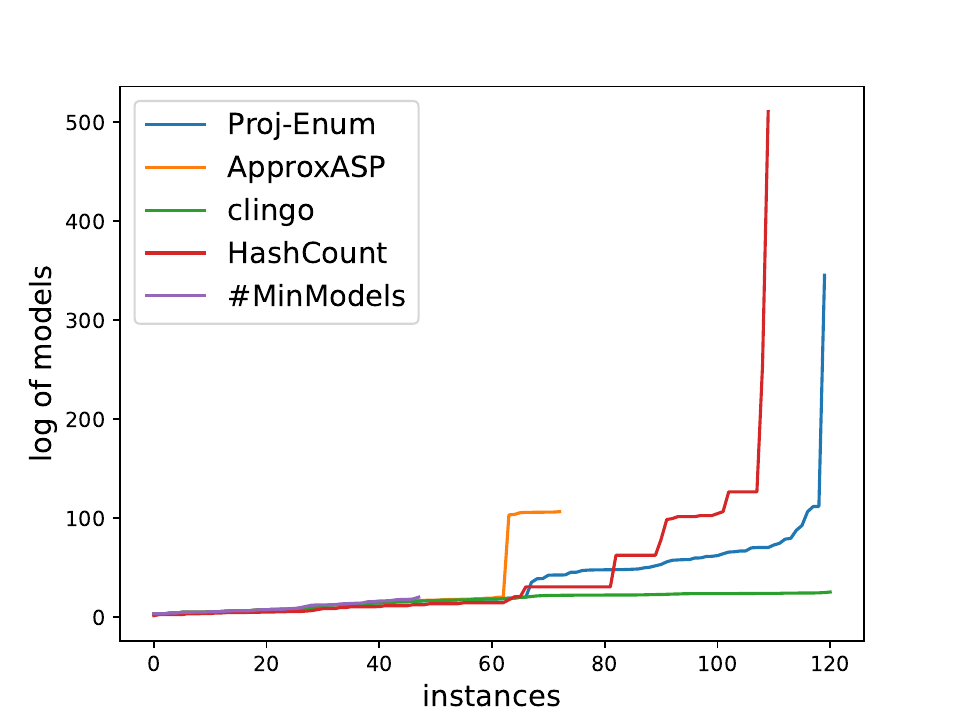}
            \caption{Model counting benchmark
            }
        \label{fig:compare_minimum_model_count}
        \end{subfigure}
        \begin{subfigure}[t]{0.48\textwidth}
            \centering
            \includegraphics[width=0.8\linewidth]{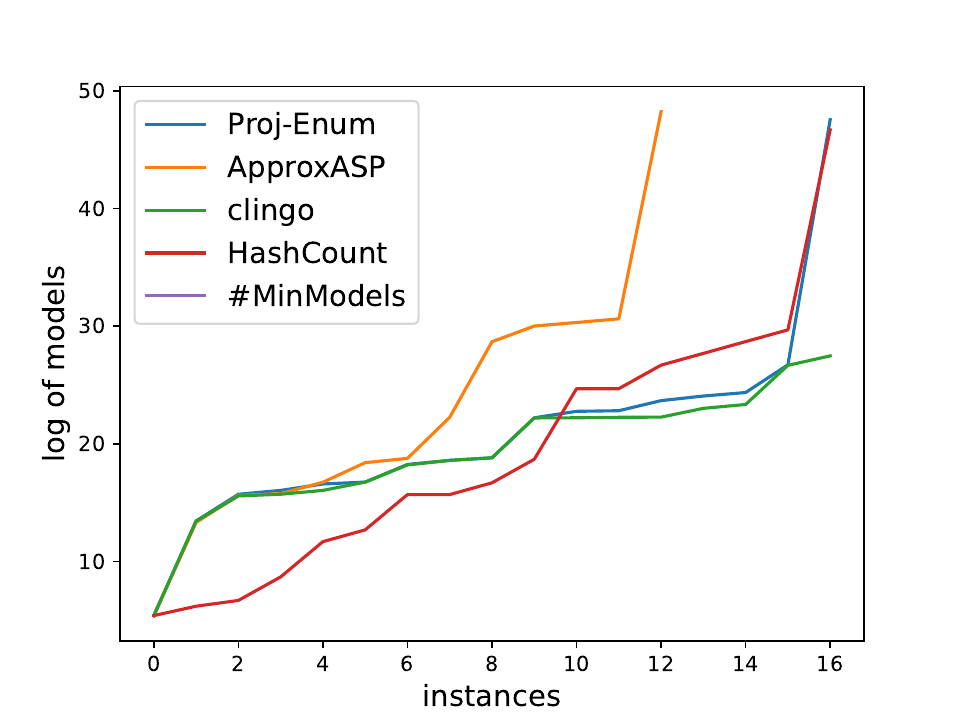}
            \caption{Minimal generators benchmark}
        \label{fig:compare_minimum_generator_count}
        \end{subfigure}
    \caption{The lower bounds returned by \projenum, \hashcounter, and existing minimal model counting tools. The $y$-axis show the log of the number of models.}
    \label{fig:compare_count_of_minimal_models}
\end{figure*}

\subsection{Another Performance Metric}
To facilitate the comparison of lower bounds returned by two tools, we have introduced another metric for comparative analysis. If tools $A$ and $B$ yield lower bounds $C_A$ and $C_B$ respectively, their {\em relative quality} is defined in the following manner:
\begin{align}
    r_{AB} = \frac{1 + \log{(C_A + 1)}}{1 + \log{(C_B + 1)}}    
    \label{eq:relative_quality}
\end{align}
If $r_{AB} > 1$, then the lower bound returned by tool $A$ is superior to that of tool $B$.

\paragraph{In-depth Study on \projenum and \hashcounter}
The performance of \projenum and \hashcounter contingent upon the size of the cut and independent support, respectively.
In this analysis, we explore the strengths and weaknesses of \projenum and \hashcounter by measuring their relative quality, as defined in~\Cref{eq:relative_quality}, across various sizes of cuts and independent supports, respectively.
This comparative analysis is visually represented in~\Cref{fig:ablation_study_across}, where clingo serves as the reference baseline.

In the graphical representations, each point $(x,y)$ corresponds to an instance where 
for the size of cut (independent support resp.) is $x$ and the prototype \projenum (\hashcounter resp.) achieves a relative quality of $y$. 
A relative quality exceeding $1$ indicates that the lower bound returned by \projenum or \hashcounter surpasses that of clingo. 
These plots reveal that \projenum tends to perform well with smaller cut sizes, while \hashcounter demonstrates better performance across a range from small to medium sizes of independent support.
\begin{figure*}[h!]
    \centering
        \begin{subfigure}[t]{0.48\textwidth}
            \centering
            \includegraphics[width=0.8\linewidth]{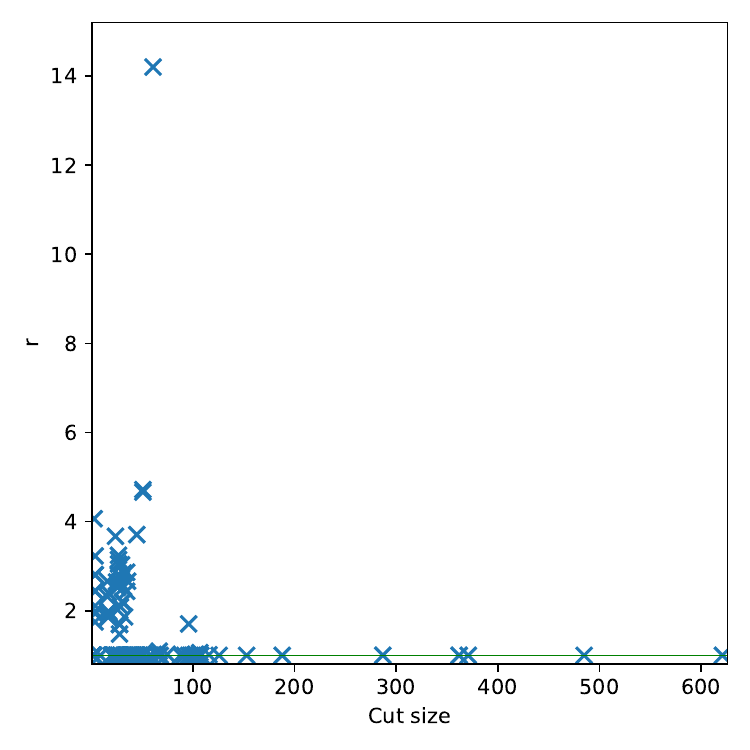}
            \caption{The relative quality of \projenum with varying cut size.
            }
        \label{fig:cut_d1_min_gen}
        \end{subfigure}
        \begin{subfigure}[t]{0.48\textwidth}
            \centering
            \includegraphics[width=0.8\linewidth]{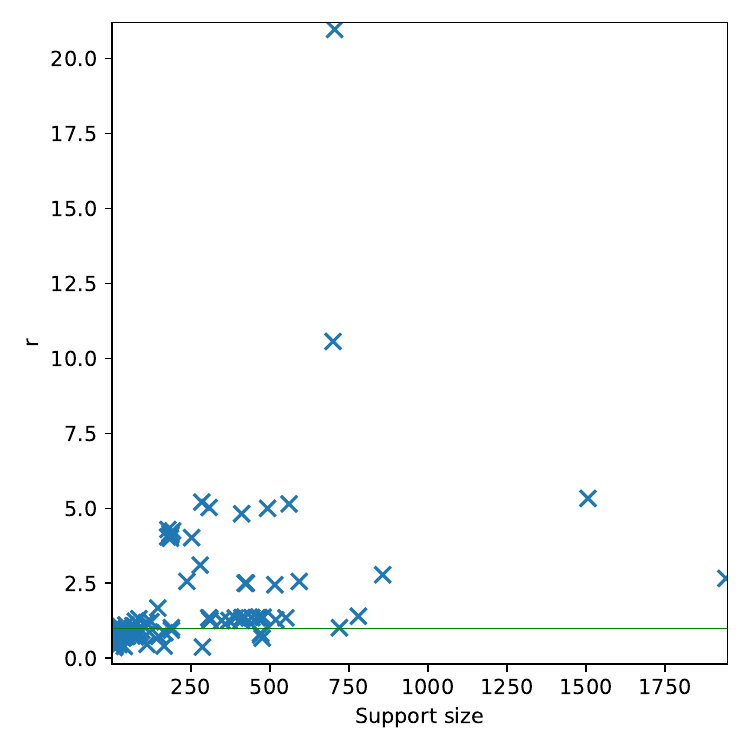}
            \caption{The relative quality of \hashcounter with varying independent support size.
            }
        \label{fig:support_mbound_min_gen}
        \end{subfigure}
    \caption{The relative quality of \projenum and \hashcounter vis-a-vis different cut and independent support size, where clingo is used as the reference baseline.
    The horizontal line is drawn across $r = 1$.}
    \label{fig:ablation_study_across}
\end{figure*}

\section{Conclusion}
\label{sec:conclusion}
This paper introduces two innovative methods for computing a lower bound on the number of minimal models. 
Our first method, \projenum, leverages knowledge compilation techniques to provide improved lower bounds for easily decomposable formulas. 
The second method, \hashcounter, utilizes recent advancements in ASP+XOR reasoning systems and demonstrates performance that varies with the size of the independent support. 
Our proposed methods exploit the expressive power of ASP semantics and robustness of well-engineered ASP systems. 
Looking forward, our research will focus on counting projected minimal models. 
We also plan to explore the counting of {\em minimal correction subsets}, which are closely related to minimal models.

\bibliography{main}

\begin{thebibliography}{}

\bibitem[Agrawal et~al., 2021]{APM2021}
{\sc Agrawal, D.}, {\sc Pote, Y.}, {\sc and} {\sc Meel, K.~S.}
\newblock Partition function estimation: A quantitative study.
\newblock {\em arXiv preprint arXiv:2105.11132} 2021.

\bibitem[Alrabbaa et~al., 2018]{ARS2018}
{\sc Alrabbaa, C.}, {\sc Rudolph, S.}, {\sc and} {\sc Schweizer, L.}
\newblock Faceted answer-set navigation.
\newblock In {\em RuleML+RR} 2018, pp. 211--225. Springer.

\bibitem[Alviano, 2017]{alviano2017}
{\sc Alviano, M.}
\newblock Model enumeration in propositional circumscription via unsatisfiable
  core analysis.
\newblock {\em TPLP}, 17(5-6):708--725 2017.

\bibitem[Alviano et~al., 2022]{ADFPR2022}
{\sc Alviano, M.}, {\sc Dodaro, C.}, {\sc Fiorentino, S.}, {\sc Previti, A.},
  {\sc and} {\sc Ricca, F.}
\newblock Enumeration of minimal models and {MUS}es in {WASP}.
\newblock In {\em LPNMR} 2022, pp. 29--42. Springer.

\bibitem[Angiulli et~al., 2014]{ABFL2017}
{\sc Angiulli, F.}, {\sc Ben-Eliyahu, R.}, {\sc Fassetti, F.}, {\sc and} {\sc
  Palopoli, L.}
\newblock On the tractability of minimal model computation for some {CNF}
  theories.
\newblock {\em Artificial Intelligence}, 210:56--77 2014.

\bibitem[Angiulli et~al., 2022]{ABFL2022}
{\sc Angiulli, F.}, {\sc Ben-Eliyahu, R.}, {\sc Fassetti, F.}, {\sc and} {\sc
  Palopoli, L.}
\newblock Graph-based construction of minimal models.
\newblock {\em Artificial Intelligence}, 313:103754 2022.

\bibitem[Ben-Eliyahu, 2005]{Rachel2005}
{\sc Ben-Eliyahu, R.}
\newblock An incremental algorithm for generating all minimal models.
\newblock {\em Artificial Intelligence}, 169(1):1--22 2005.

\bibitem[Ben-Eliyahu and Dechter, 1994]{BD1994}
{\sc Ben-Eliyahu, R.} {\sc and} {\sc Dechter, R.}
\newblock Propositional semantics for disjunctive logic programs.
\newblock {\em Annals of Mathematics and Artificial intelligence}, 12:53--87
  1994.

\bibitem[Ben-Eliyahu and Dechter, 1996]{BD1996}
{\sc Ben-Eliyahu, R.} {\sc and} {\sc Dechter, R.}
\newblock On computing minimal models.
\newblock {\em Annals of Mathematics and Artificial Intelligence}, 18(1):3--27
  1996.

\bibitem[Ben-Eliyahu and Palopoli, 1997]{RP1997}
{\sc Ben-Eliyahu, R.} {\sc and} {\sc Palopoli, L.}
\newblock Reasoning with minimal models: Efficient algorithms and applications.
\newblock {\em Artificial Intelligence}, 96(2):421--449 1997.

\bibitem[Bozzano et~al., 2022]{BCGJK2022}
{\sc Bozzano, M.}, {\sc Cimatti, A.}, {\sc Griggio, A.}, {\sc Jon{\'a}{\v{s}},
  M.}, {\sc and} {\sc Kimberly, G.}
\newblock Analysis of cyclic fault propagation via {ASP}.
\newblock In {\em LPNMR} 2022, pp. 470--483. Springer.

\bibitem[Brewka et~al., 2015]{BDRS2015}
{\sc Brewka, G.}, {\sc Delgrande, J.}, {\sc Romero, J.}, {\sc and} {\sc Schaub,
  T.}
\newblock asprin: Customizing answer set preferences without a headache.
\newblock In {\em AAAI} 2015, volume~29.

\bibitem[Cadoli, 1992a]{Cadoli1992b}
{\sc Cadoli, M.}
\newblock The complexity of model checking for circumscriptive formulae.
\newblock {\em Information Processing Letters}, 44(3):113--118 1992a.

\bibitem[Cadoli, 1992b]{Cadoli1992a}
{\sc Cadoli, M.}
\newblock On the complexity of model finding for nonmonotonic propositional
  logics.
\newblock In {\em Italian Conference on Theoretical Computer Science} 1992b,
  pp. 125--139.

\bibitem[Chakraborty et~al., 2013]{CMV2013}
{\sc Chakraborty, S.}, {\sc Meel, K.~S.}, {\sc and} {\sc Vardi, M.~Y.}
\newblock A scalable approximate model counter.
\newblock In {\em CP} 2013, pp. 200--216. Springer.

\bibitem[Chakraborty et~al., 2016]{CMV2016}
{\sc Chakraborty, S.}, {\sc Meel, K.~S.}, {\sc and} {\sc Vardi, M.~Y.}
\newblock Algorithmic improvements in approximate counting for probabilistic
  inference: From linear to logarithmic {SAT} calls.
\newblock In {\em IJCAI} 2016, pp. 3569--3576.

\bibitem[Clark, 1978]{clark1978}
{\sc Clark, K.~L.}
\newblock Negation as failure.
\newblock {\em Logic and data bases}, pp. 293--322 1978.

\bibitem[Darwiche, 2004]{Darwiche2004}
{\sc Darwiche, A.}
\newblock New advances in compiling {CNF} to decomposable negation normal form.
\newblock In {\em ECAI} 2004, pp. 328--332. Citeseer.

\bibitem[De~Kleer et~al., 1992]{DMR1992}
{\sc De~Kleer, J.}, {\sc Mackworth, A.~K.}, {\sc and} {\sc Reiter, R.}
\newblock Characterizing diagnoses and systems.
\newblock {\em Artificial intelligence}, 56(2-3):197--222 1992.

\bibitem[Durand and Hermann, 2008]{DH2008}
{\sc Durand, A.} {\sc and} {\sc Hermann, M.}
\newblock On the counting complexity of propositional circumscription.
\newblock {\em Information Processing Letters}, 106(4):164--170 2008.

\bibitem[Eiter and Gottlob, 1993]{EG1993}
{\sc Eiter, T.} {\sc and} {\sc Gottlob, G.}
\newblock Propositional circumscription and extended closed-world reasoning are
  $\prod^{\p}_2$-complete.
\newblock {\em Theoretical Computer Science}, 114(2):231--245 1993.

\bibitem[Faber et~al., 2016]{FVCG2016}
{\sc Faber, W.}, {\sc Vallati, M.}, {\sc Cerutti, F.}, {\sc and} {\sc Giacomin,
  M.}
\newblock Solving set optimization problems by cardinality optimization with an
  application to argumentation 2016.

\bibitem[Fichte et~al., 2022]{FGR2022}
{\sc Fichte, J.~K.}, {\sc Gaggl, S.~A.}, {\sc and} {\sc Rusovac, D.}
\newblock Rushing and strolling among answer sets--navigation made easy.
\newblock In {\em AAAI} 2022, volume~36, pp. 5651--5659.

\bibitem[Fichte et~al., 2021]{FHH2020}
{\sc Fichte, J.~K.}, {\sc Hecher, M.}, {\sc and} {\sc Hamiti, F.}
\newblock The model counting competition 2020.
\newblock {\em Journal of Experimental Algorithmics (JEA)}, 26:1--26 2021.

\bibitem[Gebser et~al., 2013]{GKRR2013}
{\sc Gebser, M.}, {\sc Kaufmann, B.}, {\sc Romero, J.}, {\sc Otero, R.}, {\sc
  Schaub, T.}, {\sc and} {\sc Wanko, P.}
\newblock Domain-specific heuristics in answer set programming.
\newblock In {\em AAAI} 2013, volume~27, pp. 350--356.

\bibitem[Gebser et~al., 2012]{GKS2012}
{\sc Gebser, M.}, {\sc Kaufmann, B.}, {\sc and} {\sc Schaub, T.}
\newblock Conflict-driven answer set solving: From theory to practice.
\newblock {\em Artificial Intelligence}, 187:52--89 2012.

\bibitem[Gelfond and Lifschitz, 1991]{GL1991}
{\sc Gelfond, M.} {\sc and} {\sc Lifschitz, V.}
\newblock Classical negation in logic programs and disjunctive databases.
\newblock {\em New generation computing}, 9:365--385 1991.

\bibitem[Gomes et~al., 2007]{GSS2007}
{\sc Gomes, C.~P.}, {\sc Hoffmann, J.}, {\sc Sabharwal, A.}, {\sc and} {\sc
  Selman, B.}
\newblock From sampling to model counting.
\newblock In {\em IJCAI} 2007, volume 2007, pp. 2293--2299.

\bibitem[Gomes et~al., 2006a]{GSS2006}
{\sc Gomes, C.~P.}, {\sc Sabharwal, A.}, {\sc and} {\sc Selman, B.}
\newblock Model counting: A new strategy for obtaining good bounds.
\newblock In {\em AAAI} 2006a, volume~10, pp. 1597538--1597548.

\bibitem[Gomes et~al., 2006b]{GSS2006a}
{\sc Gomes, C.~P.}, {\sc Sabharwal, A.}, {\sc and} {\sc Selman, B.}
\newblock Near-uniform sampling of combinatorial spaces using {XOR}
  constraints.
\newblock {\em NIPS}, 19 2006b.

\bibitem[Gomes et~al., 2021]{GSS2021}
{\sc Gomes, C.~P.}, {\sc Sabharwal, A.}, {\sc and} {\sc Selman, B.}
\newblock Model counting.
\newblock In {\em Handbook of satisfiability} 2021, pp. 993--1014. IOS press.

\bibitem[Hamann and Strasser, 2018]{HS2018}
{\sc Hamann, M.} {\sc and} {\sc Strasser, B.}
\newblock Graph bisection with pareto optimization.
\newblock {\em Journal of Experimental Algorithmics (JEA)}, 23:1--34 2018.

\bibitem[Hunter et~al., 2008]{HS2008}
{\sc Hunter, A.}, {\sc Konieczny, S.}, {\sc and} {\sc others}.
\newblock Measuring inconsistency through minimal inconsistent sets.
\newblock {\em KR}, 8(358-366):42 2008.

\bibitem[Jabbour et~al., 2017]{JSS2017}
{\sc Jabbour, S.}, {\sc Sais, L.}, {\sc and} {\sc Salhi, Y.}
\newblock Mining top-k motifs with a {SAT}-based framework.
\newblock {\em Artificial Intelligence}, 244:30--47 2017.

\bibitem[Jannach et~al., 2016]{JSS2016}
{\sc Jannach, D.}, {\sc Schmitz, T.}, {\sc and} {\sc Shchekotykhin, K.}
\newblock Parallel model-based diagnosis on multi-core computers.
\newblock {\em JAIR}, 55:835--887 2016.

\bibitem[Kabir et~al., 2024]{KCM2024}
{\sc Kabir, M.}, {\sc Chakraborty, S.}, {\sc and} {\sc Meel, K.~S.}
\newblock Exact {ASP} counting with compact encodings.
\newblock In {\em AAAI} 2024, volume~38, pp. 10571--10580.

\bibitem[Kabir et~al., 2022]{KESHFM2022}
{\sc Kabir, M.}, {\sc Everardo, F.~O.}, {\sc Shukla, A.~K.}, {\sc Hecher, M.},
  {\sc Fichte, J.~K.}, {\sc and} {\sc Meel, K.~S.}
\newblock {ApproxASP}--a scalable approximate answer set counter.
\newblock In {\em AAAI} 2022, volume~36, pp. 5755--5764.

\bibitem[Kabir and Meel, 2023]{KM2023}
{\sc Kabir, M.} {\sc and} {\sc Meel, K.~S.}
\newblock A fast and accurate {ASP} counting based network reliability
  estimator.
\newblock In {\em LPAR} 2023, pp. 270--287.

\bibitem[Kirousis and Kolaitis, 2003]{KK2003}
{\sc Kirousis, L.~M.} {\sc and} {\sc Kolaitis, P.~G.}
\newblock The complexity of minimal satisfiability problems.
\newblock {\em Information and Computation}, 187(1):20--39 2003.

\bibitem[Korhonen and J{\"a}rvisalo, 2021]{KJ2021}
{\sc Korhonen, T.} {\sc and} {\sc J{\"a}rvisalo, M.}
\newblock Integrating tree decompositions into decision heuristics of
  propositional model counters.
\newblock In {\em CP} 2021, pp. 8--1.

\bibitem[Lagniez and Marquis, 2017]{LM2017}
{\sc Lagniez, J.-M.} {\sc and} {\sc Marquis, P.}
\newblock An improved {Decision-DNNF} compiler.
\newblock In {\em IJCAI} 2017, volume~17, pp. 667--673.

\bibitem[Lee and Lin, 2006]{LL2006}
{\sc Lee, J.} {\sc and} {\sc Lin, F.}
\newblock Loop formulas for circumscription.
\newblock {\em Artificial Intelligence}, 170(2):160--185 2006.

\bibitem[Li et~al., 2021]{LYZR2021}
{\sc Li, Z.}, {\sc Yisong, W.}, {\sc Zhongtao, X.}, {\sc and} {\sc Renyan, F.}
\newblock Computing propositional minimal models: {MiniSAT}-based approaches.
\newblock {\em Journal of Computer Research and Development}, 58(11):2515--2523
  2021.

\bibitem[Liffiton and Sakallah, 2008]{LS2008}
{\sc Liffiton, M.~H.} {\sc and} {\sc Sakallah, K.~A.}
\newblock Algorithms for computing minimal unsatisfiable subsets of
  constraints.
\newblock {\em Journal of Automated Reasoning}, 40:1--33 2008.

\bibitem[Lifschitz, 1985]{Lifschitz1985}
{\sc Lifschitz, V.}
\newblock Computing circumscription.
\newblock In {\em IJCAI} 1985, volume~85, pp. 121--127.

\bibitem[Marek and Truszczy{\'n}ski, 1999]{MT1999}
{\sc Marek, V.~W.} {\sc and} {\sc Truszczy{\'n}ski, M.}
\newblock Stable models and an alternative logic programming paradigm.
\newblock {\em The Logic Programming Paradigm: a 25-Year Perspective}, pp.
  375--398 1999.

\bibitem[Marques-Silva et~al., 2013]{MHJPB2013}
{\sc Marques-Silva, J.}, {\sc Heras, F.}, {\sc Janota, M.}, {\sc Previti, A.},
  {\sc and} {\sc Belov, A.}
\newblock On computing minimal correction subsets.
\newblock In {\em IJCAI} 2013.

\bibitem[McCarthy, 1980]{Mccarthy1980}
{\sc McCarthy, J.}
\newblock Circumscription—a form of non-monotonic reasoning.
\newblock {\em Artificial intelligence}, 13(1-2):27--39 1980.

\bibitem[Menc{\'\i}a et~al., 2015]{MPM2015}
{\sc Menc{\'\i}a, C.}, {\sc Previti, A.}, {\sc and} {\sc Marques-Silva, J.}
\newblock Literal-based {MCS} extraction.
\newblock In {\em IJCAI} 2015.

\bibitem[Minker, 1982]{Minker1982}
{\sc Minker, J.}
\newblock On indefinite databases and the closed world assumption.
\newblock In {\em CADE} 1982, pp. 292--308. Springer.

\bibitem[Padoa, 1901]{Padoa1901}
{\sc Padoa, A.}
\newblock Essai d'une théorie algébrique des nombres entiers, précédé
  d'une introduction logique à une théorie déductive quelconque.
\newblock {\em Bibliothèque du Congrès International de Philosophie},
  3:309--365 1901.

\bibitem[Reiter, 1980]{Reiter1980}
{\sc Reiter, R.}
\newblock A logic for default reasoning.
\newblock {\em Artificial intelligence}, 13(1-2):81--132 1980.

\bibitem[Salhi, 2019]{Salhi2019}
{\sc Salhi, Y.}
\newblock On enumerating all the minimal models for particular {CNF} formula
  classes.
\newblock In {\em ICAART (2)} 2019, pp. 403--410.

\bibitem[Soos et~al., 2020]{SGM2020}
{\sc Soos, M.}, {\sc Gocht, S.}, {\sc and} {\sc Meel, K.~S.}
\newblock Tinted, detached, and lazy {CNF-XOR} solving and its applications to
  counting and sampling.
\newblock In {\em CAV} 2020, pp. 463--484. Springer.

\bibitem[Soos and Meel, 2022]{SM2022}
{\sc Soos, M.} {\sc and} {\sc Meel, K.~S.}
\newblock Arjun: An efficient independent support computation technique and its
  applications to counting and sampling.
\newblock In {\em ICCAD} 2022, pp. 1--9.

\bibitem[Thimm, 2016]{Thimm2016}
{\sc Thimm, M.}
\newblock On the expressivity of inconsistency measures.
\newblock {\em Artificial Intelligence}, 234:120--151 2016.

\bibitem[Thurley, 2006]{Thurley2006}
{\sc Thurley, M.}
\newblock {sharpSAT}--counting models with advanced component caching and
  implicit {BCP}.
\newblock In {\em SAT} 2006, pp. 424--429. Springer.

\end{thebibliography}
\newpage
\appendix
\section*{Appendix}

\section{Deferred Proofs of~\Cref{sec:preliminaries}}
\paragraph{Proof of~\Cref{lemma:minimal_model_to_minimal_generator}}
\begin{proof}
`if' part proof: Let $I = \{a\textprime_1, \ldots, a\textprime_k\} \subseteq \items$ be a minimal generator of $D$ and $\cover{I, D} = \{t\textprime_1, \ldots, t\textprime_m\}$. 
We proof that $\sigma = \{p_{a\textprime_1}, \ldots, p_{a\textprime_k}, q_{t\textprime_1}, \ldots, q_{t\textprime_m}\}$ is a minimal model of $\mingen{D}$. By definition of $\mingen{D}$, $\sigma$ is a model of $\mingen{D}$. Now we show that $\sigma$ is a minimal model 
of $\mingen{D}$. We proof it by contradiction and assume that there is model $\sigma\textprime \models \mingen{D}$ and $\sigma\textprime$ is strictly smaller than $\sigma$.
As $I$ is a minimal generator, there is no $I\textprime \subset I$ such that $\cover{I, D} = \cover{I\textprime, D}$. Thus, there is at least one 
$q_{t\textprime_u} \in \sigma \setminus \sigma\textprime$, i.e., $t\textprime_u \not\in \cover{\sigma,D}$ and $t\textprime_u \in \cover{\sigma\textprime,D}$. By definition of $\mingen{D}$, 
$q_{t\textprime_u}$ occurs exactly in one clause of $\mingen{G}$ and let us denote the clause by notation $C_{t\textprime_u}$.
The literal $q_{t\textprime_u}$ is justified in minimal model $M$, thus $\forall a \in (\items \setminus I_{t\textprime_u}), p_a \not \in M$, which follows that 
the clause $C_{t\textprime_u}$ is not satisified by $\sigma\textprime$, which contradicts that $\sigma\textprime \models \mingen{D}$. 

`only if' part proof: Let $\sigma$ be a minimal model of $\mingen{D}$ and assume that $I_{\sigma} = \{a | p_a \in \sigma\}$ is not a minimal generator of $D$
i.e., there is another $I\textprime_{\sigma} \subset I_{\sigma}$ such that $\cover{I_{\sigma}, D} = \cover{I\textprime_{\sigma}, D}$.
By construction of $\mingen{D}$, $\cover{I_{\sigma}, D} \subseteq \{q_t | q_t \in \sigma\}$. As $\sigma \in \mmodel{\mingen{D}}$, 
$\forall q_t \in \sigma$ has a justification. Note that the (positive) literal $q_t$ occurs in exactly one clause $\mingen{D}$, 
which implies that $\forall a \in (\items \setminus I_t), p_a \not \in \sigma$. Thus, $\cover{I_{\sigma}, D} = \{q_t | q_t \in \sigma\}$, 
which follows that $\cover{I\textprime_{\sigma}, D} = \cover{I_{\sigma}, D} = \{q_t | q_t \in \sigma\}$. However, if $I\textprime_{\sigma} \subset I_{\sigma}$, 
then there are two possible cases: (i) either $\sigma\textprime = \{p_a | a \in I\textprime_{\sigma}\} \cup \{q_t | q_t \in \sigma\} \models \mingen{D}$, that 
contradicts that $\sigma$ is a minimal model of $\mingen{D}$, (ii) or $\sigma\textprime = \{p_a | a \in I\textprime_{\sigma}\} \cup \{q_t | q_t \in \sigma\} \not\models \mingen{D}$
that contradicts that $\cover{I_{\sigma}, D} = \cover{I\textprime_{\sigma}, D}$.
\end{proof}

\begin{lemma}
    $\sigma \in \mmodel{F}$ if and only if $\sigma \in \answer{\dlp{F}}$
\end{lemma}
    \begin{proof}
    proof of ``if'' part: Proof idea: Proof by Contradication.\\
    Assume that $\sigma \in \answer{\dlp{F}}$ but $\sigma \not\in \mmodel{F}$.
    As $\sigma \in \answer{\dlp{F}}$, then $\sigma \models F$. Thus we only proof that there is no 
    model $\sigma\textprime \subset \sigma$ such that $\sigma\textprime \models F$.
    For purpose of the contradiction, assume that there is a model $\sigma\textprime \subset \sigma$ and 
    $\sigma\textprime \models F$. By consturction of program $\dlp{F}$, as $\sigma\textprime \models F$, $\sigma\textprime \models \dlp{F}$.
    Note that the reduct of $\dlp{F}$ w.r.t. $\sigma$ and $\sigma\textprime$ are same, more specifically, 
    $\dlp{F}^{\sigma} = \dlp{F}^{\sigma\textprime} = \dlp{F}$ because there is no default negation in 
    $\dlp{F}$. Thus, $\sigma\textprime \subset \sigma$ and $\sigma\textprime \models \dlp{F}^{\sigma}$, 
    which contradicts that $\sigma$ is an answer set of $\dlp{F}$.   
    
    proof of ``only if'' part: 
    The proof is trivial. If $\sigma \in \mmodel{F}$ and $\sigma \not\in \answer{\dlp{F}}$, then 
    there is another $\sigma\textprime \subset \sigma$ and $\sigma\textprime \models \dlp{F}$, which contradicts
    that $\sigma$ is a minimal model because $\sigma\textprime \models F$.   
\end{proof}

\section{Theoretical Analysis of \hybrid}
\paragraph{Analysis of \projenum: Proof of~\Cref{theorem:decomposition}}
\begin{proof}
The proof consists on the following two parts:
\begin{enumerate}
\item for each $\sigma \in \mmodel{F}$, there is a sequence of $\sigma_i$, $i \in [1,k]$ such that $\sigma_i \in \mmodel{F}_{\downarrow V_i}$
\item for each $\sigma_i \in \mmodel{F}_{\downarrow V_i}$ and $i \in [1,k]$, then $\bigcup_{i} \sigma_i \in \mmodel{F}$
\end{enumerate}
Proof of \textbf{1}: Proof of this part is trivial. For each $\sigma \in \mmodel{F}$, it follows that $\sigma_i = \sigma_{\downarrow V_i}$.
    
\noindent Proof of \textbf{2}: The variable set of each $\sigma_i$ is distinct or $\forall i \neq j, V_i \cap V_j = \emptyset$. For each $\sigma_i \in \mmodel{F}_{\downarrow V_i}$, there exists a $\sigma \in \mmodel{F}$, where $\sigma_i = \sigma_{\downarrow V_i}$. The $\bigcup_{i} \sigma_i$ assigns the variable of $V$ and the union of $\sigma_i$ or $\sigma = \bigcup_{i} \sigma_i$ is a minimal model of $F$
\end{proof}
The correctness of~\Cref{alg:projected_enumeration} can be established as follows:
\begin{lemma}
    \label{lemma:proof_of_enumproj}
    For Boolean formula $F$ and a cut $\cut$, the minimal models of $F$ can be computed as follows:
        $\mmodel{F} = \bigcup_{\tau \in 2^{\cut}} \projenu{F, \tau, \var{F}}$.
\end{lemma}
\begin{proof}
    ~\Cref{theorem:decomposition} demonstrates that minimal models can be computed through component decompositions. 
    By taking the union over $\tau \in 2^{\cut}$, we iterate over all possible assignments of $\cut$. Consequently, \projenum 
    algorithm computes all minimal models of $F$ by conditioning over all possible assignments over $\cut$. 
    Therefore, the algorithm is correct. 
\end{proof}

\paragraph{Analysis of \hashcounter.}
We adopt the following notation: $\sstar = \log_{2}{\Card{\mmodel{F}}}$. 
Each minimal model $\sigma$ of $F$ is an assignment over $\var{F}$, 
and according to the definition of random XOR constraint~\citep{GSS2006a}, $\sigma$ satisifies a random XOR constraint with probability of $\sfrac{1}{2}$.
Due to uniformity and randomness of XOR constraints, each minimal model of $F$ satisfies $m$ random and uniform XOR constraints with probability of $\sfrac{1}{2^m}$.
In our theoretical analysis, we apply the Markov inequality: if $Y$ is a non-negative random variable, then $\pr{Y \geq a} \leq \frac{\mathbb{E}[Y]}{a}$, where $a > 0$.
\begin{lemma}
    \label{lemma:prob_of_each_s}
    For arbitrary $s$, $\pr{\Card{\mmodel{F^s} \geq 1}} \leq \frac{2^{\sstar}}{2^s}$.
\end{lemma}
\begin{proof}
    For each $\sigma \in \mmodel{F}$, we define a random variable $Y_\sigma \in \{0,1\}$ and $Y_\sigma = 1$ indicates that $\sigma$ satisifies the first $s$ XOR constraints $Q^1, \ldots, Q^{s}$, otherwise, $Y_\sigma = 0$.
    The random variable $Y$ is the summation, $Y = \sum_{\sigma \in \mmodel{F}} Y_{\sigma}$. 
    The expected value of $Y$ can be calculated as $\mathbb{E}(Y) = \sum_{\sigma \in \mmodel{F}} \mathbb{E}(Y_\sigma)$.

    Due to the nature of random and uniform XOR constraints, each minimal model $\sigma \in \mmodel{F}$ satisfies all $s$ XOR constraints with probability $\frac{1}{2^s}$.
    It follows that the expected value $\mathbb{E}(Y_\sigma) = \frac{1}{2^s}$, and the expected value of $Y$ is $\mathbb{E}(Y) = \frac{\Card{\mmodel{F}}}{2^s} = \frac{2^{\sstar}}{2^s}$. According to the Markov inequality:
    \begin{align*}
        \pr{\Card{\mmodel{F^s} \geq 1}} &\leq \frac{2^{\sstar}}{2^s}
    \end{align*}
\end{proof}
\begin{lemma}
    \label{lemma:guarantee_minimum}
    Given a formula $F$ and confidence $\delta$, if $\approxfunction(F, \delta)$ returns $2^{s - \alpha}$, then $\pr{2^{s - \alpha} \leq \Card{\mmodel{F}}} \geq 1 - \delta$ 
\end{lemma}
\begin{proof}
    Given a input Boolean formula $F$ and for each $m \in [1, \Card{\support} - 1]$, we denote the following two events: $I_m$ denotes the event that \Cref{alg:lower_bound_on_minimal_models_counting} invokes $\mmodel{F^m}$ and 
    $E_m$ denotes the event that $\Card{\mmodel{F^m}} \geq 1$. The algorithm $\approxfunction(F, \delta)$ returns an incorrect bound when $s - \alpha > \sstar$ and let use the notation $\error$ to denote 
    that $\approxfunction(F, \delta)$ returns an incorrect bound or $\approxfunction(F, \delta) > 2^{\sstar}$. The upper bound of $\pr{\error}$ can be calculated as follows:
    \begin{align*}
        \pr{\error} &= \pr{\text{$\approxfunction(F, \delta)$ returns $2^{s - \alpha}$ and } s > \sstar + \alpha}\\
        &\leq \sum_{s > \sstar + \alpha} \pr{I_s \cap E_s} \leq \sum_{s > \sstar + \alpha} \pr{E_s}\\
        &\leq \sum_{s > \sstar + \alpha} \pr{\Card{\mmodel{F^s}} \geq 1} \leq \sum_{s > \sstar + \alpha} \frac{2^{\sstar}}{2^s} &\qquad & \text{According to~\Cref{lemma:prob_of_each_s}}\\
        &\leq \frac{1}{2^{\alpha}} \times 2 \leq 2^{1-\alpha} \leq 2^{\log_2{\delta}} \leq \delta
    \end{align*}
    Thus, $\pr{\text{$\approxfunction(F, \delta)$ returns $2^{s - \alpha}$ and } s \leq \sstar + \alpha} \geq 1 - \delta$.
\end{proof}
\paragraph{Analysis of \Cref{thm:main_theorem}}
\begin{proof}
    The proof consists of two cases. 
    
    \noindent (i) When \hybrid calls \projenum (if $\Card{\cut} \leq 50$). In the case, the proof follows~\Cref{lemma:proof_of_enumproj}. 
    
    \noindent (ii) \hybrid calls \hashcounter. In the case, the proof follows~\Cref{lemma:guarantee_minimum}. 
\end{proof}

\end{document}